\documentclass{article}

\usepackage[english]{babel}
\usepackage{amsmath}
\usepackage{amsfonts}
\usepackage{amsthm}
\usepackage{caption}
\usepackage{bbm}
\usepackage[utf8]{inputenc}
\usepackage[T1]{fontenc}
\usepackage{inconsolata}

\usepackage{authblk}

\usepackage[sorting=none,backend=bibtex]{biblatex}
\bibliography{../latexCommon/biblio}

\newcommand{\T}{\mathcal{T}}
\newcommand{\VV}{\mathcal{V}}

\newcommand{\F}{\mathcal{F}}

\newcommand{\CC}{C\nolinebreak\hspace{-.05em}\raisebox{.4ex}{\tiny\bf +}\nolinebreak\hspace{-.10em}\raisebox{.4ex}{\tiny\bf +}}
\def\CC{{C\nolinebreak[4]\hspace{-.05em}\raisebox{.4ex}{\tiny\bf ++}}}
\newcommand{\dP}{\mathcal{P}}
\newcommand{\D}{\partial}
\newcommand{\DD}{\mathcal{D}}
\newcommand{\sumd}{\tau}

\usepackage{listings}
\usepackage{xcolor}
\newtheorem{definition}{Definition}[section]

\newtheorem{theorem}{Theorem}[section]

\newtheorem{corollary}{Corollary}[section]

\usepackage[hidelinks]{hyperref}
\usepackage{epigraph}
\usepackage{etoolbox}

\usepackage{tikz}
\usepackage{tikz-cd}
\usetikzlibrary{shapes,arrows}
\usetikzlibrary{positioning}
\usepackage{algorithm}
\usepackage{algpseudocode}
\usepackage{pifont}
\title{Automatic Differentiation: a look through Tensor and Operational Calculus}
\author[1]{Žiga Sajovic}
\affil[1]{\footnotesize XLAB d.o.o.\authorcr ziga.sajovic@xlab.si}
\date{}
\begin{document}

\maketitle

\begin{abstract}
  In this paper we take a look at \emph{Automatic Differentiation} through the eyes of \emph{Tensor} and \emph{Operational Calculus}. This work is best consumed as supplementary material for learning tensor and operational calculus by those already familiar with automatic differentiation. To that purpose, we provide a simple implementation of automatic differentiation, where the steps taken are explained in the language of tensor and operational calculus.
\end{abstract} 

\section{Introduction}\label{sec:introduction}

In this paper we take a look at \emph{Automatic Differentiation} through the eyes of \emph{Tensor} and \emph{Operational Calculus}. This work is best consumed as supplementary material for learning tensor and operational calculus by those already familiar with automatic differentiation. We first note the difference between the two concepts: \emph{automatic differentiation} is a name of a collection of techniques used for efficient computation of derivatives of differentiable programs, while \emph{operational calculus} is a technique that enables reasoning about analytic properties of differentiable programs through an algebra of higher order constructs. Thus, as rates of change (derivatives) are analytic properties, and automatic differentiation is used to compute them, we can explain automatic differentiation through the language of tensor and operational calculus. Note that this work will only cover as much tensor and operational calculus as needed for explanations of the implementation they are accompanying. This means that we will not delve into the many uses of operational calculus for differentiable programming that scope beyond describing automatic differentiation. As such, employment of operational calculus for reasoning about differentiable programs, ie. analysing fractional iterations, derivations of differentiable fixed point combinators, etc., will not be covered in this work (but can be found in \cite{OperationalCalculus}).

We provide a simple implementation of the memory space $\VV$ and its expansion to $\VV\otimes T(\VV^*)$, where $T(\VV^*)$ is the tensor algebra of its dual, serving by itself as an \emph{algebra of programs} \cite[Definition~4.1]{OperationalCalculus}. We than define the concept of a \emph{differentiable programming space}, implement the operator $\tau_n$, that increases the order of a differentiable programming space \cite[Proposition~5.1]{OperationalCalculus} and enable \emph{differentiable derivatives}. The operator $\tau_n$ is than used to implement an illustrative example of a differentiable programming space of arbitrary order.
Complete source code can be found on GitHub \cite{dCpp}.

\section{Virtual memory space}\label{sec:virtualMemory}

We model an element of the virtual memory space $v\in\VV_n$ \cite[Definition~4.1]{OperationalCalculus}
\begin{eqnarray}
\VV_{n}=\VV_{n-1}\oplus(V_{n-1}\otimes\VV^*) \label{eq:V_n}
\end{eqnarray}
\begin{equation}
\VV_{n}=\VV\oplus\VV\otimes\VV^*\oplus\cdots\oplus\VV\otimes\VV^{*n\otimes} \implies\VV_n=\VV\otimes T(\VV^*)
\end{equation}
with the class $var$.

\begin{lstlisting}
template<class V>
class var
{
    public:
    	int order;
        V id;
        std::shared_ptr<std::map<var*, var> >* dTau;

        var();
        var(V id);
        var(const var& other);
        ~var();
        void init(int order);
        var d(var* dVar);
        //declerations of algebraic operations
        //declerations of order logic
};
\end{lstlisting}

The virtual memory space $\VV_n$ is the tensor product of the virtual memory $\VV$ with the tensor algebra of its dual. Tensor products of the virtual memory $\VV$ with its dual are modeled using maps, denoted by \emph{dTau}. The address \emph{var*} stands for the component of $v\in\VV_{n-1}$ with which the tensor product of the component of $v^*\in\VV^*$ was computed to generate $v\in\VV_n$ in equation \eqref{eq:V_n}. This depth is contained in the \emph{int order}.

 \begin{theorem}
 An instance of the class $var$ is an element of the virtual memory $\VV_n$.
  \begin{equation}
  var\in\VV_n
  \end{equation}
 \end{theorem}
 \begin{proof}
 \begin{equation}
  id\in \VV\land  dTau\in \VV_{n-1}
  \end{equation}
  $$\land$$
  \begin{equation}
  var=id\oplus dTau\implies var\in\VV_n
  \end{equation}
 \end{proof}

\subsection{Initialization}

A constant element $v_0$ of the virtual memory is an element of $\VV_0=\VV<\VV_n$. We initialize an element to be $n$-differentiable, by mapping
\begin{equation}
init:\VV\times\mathbb{N}\to\VV_n,
\end{equation}
\begin{equation}
v_0.init(n)=v_n\in\VV_n
\end{equation}
The image $v_n$ is an element of $\VV_1=\VV\otimes\VV^*$ with a natural inclusion in $\VV_n$.
\begin{equation}
v_n=(v_0\in\VV)+(\delta^i_j\in\VV\otimes\VV^ {*\otimes})+\sum\limits_{i=2}^n(0\in\VV\otimes\VV^ {*i\otimes}),
\end{equation}
where $\delta^i_j$ is the Kronecker delta.

\subsection{Algebra over a field}\label{sec:Algebra}

 An algebra over a field is a vector space equipped with a bilinear product. Thus, an algebra is an algebraic structure which consists of a set, together with operations of multiplication, addition, and scalar multiplication by elements of the underlying field. \cite[p.~3]{Algebra}
 
We will construct the algebra by explicitly mimicking the application of a direct sum of operators $\sumd_n$ mapping $\dP_0\oplus\dP_0\to\dP_n$ to the maps of scalar multiplication, addition and a bilinear product. 
Their differentiable compositions are modeled by mimicking projections of the operator of program composition \cite[Theorem~5.3]{OperationalCalculus}
 \begin{equation}\label{eq:kom}
   \exp(\D_fe^{h\D_g}): \dP\to \dP\to\dP_\infty.
   \end{equation}
and fixing one of the mappings \cite[eq.~44,~45]{OperationalCalculus}. If the readers would prefer to implement the algebra through the application of operators to appropriate mappings
\footnote{Ex.: the map representing addition, or the bilinear product map.},
they can do so using the operators we construct in section \ref{sec:operators}.

\subsubsection{Vector space over a field $K$}\label{sec:vectorSpace}
 
As part of the construction of an algebra, we construct a vector space $\VV_n$ over a field $K$. A vector space is a collection of objects that can be added together and multiplied with elements of the underlying field $K$.

We begin with scalar multiplication.

\begin{lstlisting}
template<class K>
var var::operator*(K n)const{
    var out;
    out.id = this->id*n;
    for_each_copy(..., mul_make_pair<pair<var*, var> >, n);
    return out;
}

template<class K>
var var::operator/(K n)const {...};
\end{lstlisting}
Scalar multiplication and its convenient inverse employ the function \emph{for\_each\_copy}, applying the provided operation 
\begin{lstlisting}
template<class V, class K>
V mul_make_pair(V v, K n) {
    return std::make_pair(v.first, v.second * n);
}
\end{lstlisting}
to each one of the components of $this$ and storing the result in $out.dTau$.

Vector addition by component is implemented by 
\begin{lstlisting}
var var::operator+(const var& v)const{
    var out;
    out.id = this->id+v.id;
    merge_apply(..., sum_pairs<pair<var*, var> >);
    return out;
}
\end{lstlisting}
Vector addition by component employs the function \emph{merge\_apply}, applying the provided function \emph{sum\_pairs}
\begin{lstlisting}
template<class V>
T sum_pairs(V v1, V v2) {
    return std::make_pair(v1.first, v1.second + v2.second);
}
\end{lstlisting}
to corresponding components, storing the result in $out.dTau$.

\begin{theorem}
Class $var$ models a vector space over a field $K$.
\end{theorem}
\begin{proof}
By implementations of addition by components and multiplication with a scalar $k\in K$, the axioms of the vector space are satisfied.
\end{proof}

\subsubsection{Algebra over a field $K$}\label{sec:algebra}

To construct an algebra over a field $K$, we equip the vector space $\VV_n$ with a bilinear product by components.

\begin{lstlisting}
var var::operator*(const var& v)const{
    var out;
    out.id = this->real*v.id;
    out.order = this->order<v.order?this->order:v.order;
    if(out.order>0){
        map<int, double> tmp1;
        map<int, double> tmp2;
        for_each_copy(..., mul_make_pair<pair<var*, var> >, 
          v.reduce());
        for_each_copy(..., mul_make_pair<pair<var*, var> >, 
          this->reduce());
        merge_apply(..., sum_pairs<pair<var*, var> >);
    }
    return out;
}
\end{lstlisting}
The employed functions have been explained at previously use. The function \emph{reduce} makes a shallow copy of $this$, while reducing the order of the returned copy (making it one time less differentiable). This bilinear product contains Leibniz rule within its structure, as the projection of the operator $\exp(\D_fe^{h\D_g})(g): \dP\to\dP_\infty(g)$ to the unit $n$-cube was applied to the algebra.

\begin{theorem}
Class $var$ models an algebra over a field $K$.
\end{theorem}
\begin{proof}
By the implementation of a bilinear product by components the axioms of an algebra over a field are satisfied.
\end{proof}

We may now trivially implement the remaining operators existing in the algebra.

\begin{lstlisting}
var var::operator-(const var& v)const{
    return *this+(-1)*v;
}

var var::operator/(const var& v)const{
    return *this*(v^(-1));
}
\end{lstlisting}
Similarly, the following operators can be assumed to be generated by the existing algebra.

\begin{lstlisting}
var operator*(double n)const;
var operator+(double n)const;
var operator-(double n)const;
var operator/(double n)const;
var operator^(double n) const;
var& operator=(const var& v);
var& operator+=(const var& v);
var& operator-=(const var& v);
var& operator*=(const var& v);
var& operator/=(const var& v);
var& operator*=(double n);
var& operator/=(double n);
var& operator+=(double n);
var& operator-=(double n);
var operator*(double n, const var& v);
var operator+(double n, const var& v);
var operator-(double n, const var& v);
var operator/(double n, const var& v);
var operator^(double n, const var& v);
\end{lstlisting} 

Order logic is trivially implemented, by mapping
\begin{equation}
\VV.id\times\VV.id\to\{0,1\}.
\end{equation}
\begin{lstlisting}
bool operator==(const var& v)const;
bool operator!=(const var& v)const;
bool operator<(const var& v)const;
bool operator<=(const var& v)const;
bool operator>(const var& v)const;
bool operator>=(const var& v)const;
\end{lstlisting}

\section{Differentiable programming space}\label{sec:differentiableProgSpace}

With the algebra over $\VV_n$ implemented, we turn towards constructing a differentiable programming space, from the programming space
              
\begin{equation}\label{eq:cc}
\CC:\VV\to\VV.
\end{equation}

\begin{definition}[Differentiable programming space]\label{def:dP}
  A \emph{differentiable programming space} $\dP_0$ is any subspace of $\F_0:\VV\to\VV$ such that
  \begin{equation}\label{eq:P}
  \D\dP_0\subset\dP_0\otimes T(V^*)
  \end{equation}
  The space $\dP_n<\F_n:\VV\to\VV_n$, spanned by $\{\D^k\dP_0;\quad 0\le k\le n\}$ over $K$, is called a differentiable programming space of order $n$. When all elements of $\dP_0$ are analytic, we denote $\dP_0$ as an \emph{analytic programming space}. \cite[Definition~4.2]{OperationalCalculus}
 \end{definition}
              
\begin{theorem}[Infinite differentiability]\label{izr:P}
  Any differentiable programming space $\dP_0$ is an
  infinitely differentiable programming space, meaning that
  \begin{equation}\label{eq:P_n}
      \D^k\dP_0\subset\dP_0\otimes T(V^*)
    \end{equation}
for any $k\in\mathbb{N}$. \cite[Theorem~4.1]{OperationalCalculus}
\end{theorem}

It follows from Theorem \ref{izr:P} that a differentiable programming space of order $n$, $\dP_n:\VV\to\VV\otimes T(\VV^*)$ can be embedded into the tensor product space of the programming space $\dP_0$ and the space $T_n(\VV^*)$ of multi-tensors of order less than equal $n$ \cite[Corollary~4.1.1]{OperationalCalculus}. For the specific case of \eqref{eq:cc} this entails
\begin{equation}
\DD^n\CC<\dP_n\iff\DD^n\CC\subset\CC\otimes T(\VV^*)
\end{equation}
which we ensure by providing closure of \eqref{eq:cc} under $\D$. Furthermore, this means that the tuple $(\VV,\dP_0)$
\footnote{In our case $\dP_0:=\CC$ \eqref{eq:cc}.}
 - together with the structure of the tensor algebra $T(\VV^*)$ - is sufficient for constructing arbitrary differentiable programming spaces $\dP_n$ .

\subsection{Operators}\label{sec:operators}

In this subsection we provide a simple implementation of $\sumd_n$ and its pullback through an arbitrary program. This is equivalent to the projection of the operator of program composition to the unit hyper-cube. As we have already constructed the algebra of the memory $\VV_n$, such an operator will enable us to construct differentiable programming spaces of arbitrary order.

The operator of tensor series expansion \cite[Theorem~5.1]{OperationalCalculus}

\begin{equation}\label{eq:specProg}
            e^{h\D}:\dP\to\VV\to \VV\otimes \T(\VV^*),
          \end{equation}

\begin{equation}
  e^{h\D}=\sum\limits_{n=0}^{\infty}\frac{(h\D)^n}{n!}
 \end{equation}
 is evaluated at $h=1$ and projected onto the unit $N$-cube, arriving at 
\begin{equation}\label{eq:DD}
    \sumd_N = 1+\D +\D^2 +\ldots + \D^N.
  \end{equation}
The expression \eqref{eq:DD} obeys the recursive relation
\begin{equation}
      \label{eq:rekurzija}
      \sumd_{k+1} = 1+\D\sumd_k,
    \end{equation}
which can be used for constructing programming spaces of arbitrary order. Note that only explicit knowledge of $\tau_1:\dP_0\to\dP_1$ is required \cite[Proposition~5.1]{OperationalCalculus}. We perform its pullback through an arbitrary program by mimicking the operator of program composition \cite[Theorem~5.2]{OperationalCalculus}
\begin{equation}\label{eq:opKompo}
  \exp(\D_fe^{h\D_g}): \dP\to\dP\to\dP_\infty,
  \end{equation}
partially evaluating on the second map and projecting it onto the unit $N$-cube.         

In our definition of the class \emph{tau}, we let \emph{K} model the image of the identity ($1$) and \emph{dTau} model the image of $\D\tau_{k-1}$ from the expression \eqref{eq:rekurzija} after application. The former is of the type $F:K\to K$, representing the mapping from the underlying field to itself, and the latter of type $dTau:\VV\to\VV_{k-1}$, representing its differentiable derivative.  
\begin{lstlisting}
template<class dTau, class K>
class tau
{
public:
    tau();
    tau(F mapping, dTau primitive);
    ~tau();
    var operator()(const var&v);
private:
    dTau primitive;
    F mapping;
};
\end{lstlisting}        
        
\begin{lstlisting}
var tau::operator()(const var&v){
    var out;
    out.idi = mapping(v.id);
    for_each_copy(..., mul_make_pair<std::pair<var*, var> >, 
      primitive(v.reduce()));
    return out;
}
\end{lstlisting}
Note that while the class \emph{tau} represents the operator \eqref{eq:rekurzija}, its instances represent the image of the operator \eqref{eq:rekurzija}, via its (\emph{tau}) evaluation operator \emph{()}.

\subsection{Differentiable Derivatives}\label{sec:orderReduction}
 
The ability to treat the $k$-th derivative of a program, as part of a different differentiable program appears useful in many fields. To that purpose, we must be able to threat the derivative itself as a differentiable program. This is easily implemented given our simplified construction, as we are working solely with the projections of the operators to the unit $N$-cube, which eases our troubles.

\begin{theorem}\label{izr:reductionMap}
There exists a reduction of order map $\phi:\dP_n\to \dP_{n-1}$, such that the
following  diagram commutes
\begin{equation}\label{eq:reductionMap}
\begin{tikzcd}
  \dP_n \arrow{r}{\phi} \arrow{d}{\D} & 
  \dP_{n-1} \arrow{d}{\D}\\
  \dP_{n+1} \arrow{r}{\phi} & 
  \dP_{n}
\end{tikzcd}
\end{equation}
satisfying
\begin{equation}
\forall_{P_1\in\dP_0}\exists_{P_2\in\dP_0}\Big(\phi^k\circ \sumd_n(P_1)=\sumd_{n-k}(P_2)\Big)
\end{equation}
for each $n\ge 1$.
\end{theorem}  
\begin{corollary}\label{cor:extraxtDerivatives}
By Theorem \ref{izr:reductionMap}, $n$-differentiable $k$-th derivatives of a program $P\in\dP_0$ can be extracted by
\begin{equation}
^{n}P^{k\prime}=\phi^k\circ \sumd_{n+k}(P)\in\dP_n
\end{equation}
\end{corollary}

Given that we are only interested in computing (differentiable) derivatives of programs (i.e. automatic differentiation), and do not need the algebra of higher-order constructs the full treatment of operational calculus would afford us, we are working entirely on the unite $N$-cube. Hence, our reduction of order map is a simple projection to the $(N-1)$-cube.

We construct the reduction of order map $d:\dP_n\to\VV\to\dP_{n-1}$,
\begin{lstlisting}
var var::d(var* dvar){
    return (*this -> dTau.get())[dvar];
}
\end{lstlisting}
returning an $(n-1)$-differentiable derivative with respect to $v_i$. It nature is further explained by
\begin{equation}
v=v.id+\sum\limits_{\forall_i}v.d(\&v_i)\otimes dv_i \in\VV_n\implies v.d(\&v_i)\in\VV_{n-1},
\end{equation}
expressing $v\in\VV_n$ in terms of it.

\subsection{Example}

We conclude by employing the constructed operators and the tensor algebra of the virtual memory space to implement a minimal differentiable programming space \emph{dCpp}.

\begin{lstlisting}
namespace dCpp{
    var sin(const var& v);
    tau cos;
    tau e;
    tau ln;
    var cos_primitive(const var& v);
    var ln_primitive(const var& v);
    var e_primitive(const var& v);
}
\end{lstlisting}

\noindent We construct the map $sin$ explicitly for educational purposes.

\begin{lstlisting}
var dCpp::sin(const var& v){
    var out;
    out.id = std::sin(v.id);
    out.order = v.order;
    if(v.order>0){
    	for_each_copy(..., mul_make_pair<std::pair<var*, var> >,
    		dCpp::cos(v.reduce()));
    }
    return out;
}
\end{lstlisting}
Other maps are constructed by employing the operator \eqref{eq:rekurzija}.
\begin{lstlisting}
typedef var (*dTau)(var);
typedef double (*F)(double);

var dCpp::cos_primitive(const &var v){
    return (-1)*dCpp::sin(v);
}

dCpp::cos = tau<dTau, F>(std::cos, dCpp::cos_primitive);

var dCpp::ln_primitive(const &var v){
    return 1/v;
}

dCpp::ln = tau<dTau, F>(std::ln, dCpp::ln_primitive);

dCpp::e_primitive(const &var v){
    return dCpp:e(v);
}

tau dCpp::e = tau<dTau, F>(std::e, dCpp::e_primitive);

\end{lstlisting}

\section{Conclusions}

In this work we provided a simple implementation of automatic differentiation, where the steps taken were explained in the language of tensor and operational calculus. It aimed to serve as a gentle introduction into tensor and operational calculus for those already familiar with automatic differentiation.
We hope this introduction motivates the practitioner of differentiable programming to look into further uses of operational calculus, that can accommodate more than just a description of automatic differentiation presented in this work.

\printbibliography
\end{document}